\documentclass[runningheads]{llncs}
\usepackage{amssymb,url,graphicx,algorithmic,color}
\usepackage{tikz}
\usepackage{marvosym}
\newcommand{\np}{\mathsf{NP}}

\begin{document}
\title{Similarity Between Points in \\ Metric Measure Spaces}
%
%
\author{Evgeny Dantsin \textsuperscript{\Letter} \and Alexander Wolpert \textsuperscript{\Letter}}
\authorrunning{E. Dantsin \and A. Wolpert}
%
\institute{Roosevelt University, Chicago, IL, USA \\
\email{\{edantsin,awolpert\}@roosevelt.edu}}
\maketitle              
\begin{abstract}
This paper is about similarity between objects that can be represented as points in metric measure spaces. A metric measure space is a metric space that is also equipped with a measure. For example, a network with distances between its nodes and weights assigned to its nodes is a metric measure space. Given points $x$ and $y$ in different metric measure spaces or in the same space, how similar are they? A well known approach is to consider $x$ and $y$ similar if their neighborhoods are similar. For metric measure spaces, similarity between neighborhoods is well captured by the Gromov-Hausdorff-Prokhorov distance, but it is $\np$-hard to compute this distance even in quite simple cases. We propose a tractable alternative: the \emph{radial distribution distance} between the neighborhoods of $x$ and $y$. The similarity measure based on the radial distribution distance is coarser than the similarity based on the Gromov-Hausdorff-Prokhorov distance but much easier to compute. 

\end{abstract}
\section{Introduction}
\label{sec:intro}

A \emph{metric measure space} is a metric space that is also equipped with a measure. 
Such spaces play an important role in geometry, especially after Gromov's works \cite{Gro99}, and they have proven to be useful in other areas of mathematics, for example, in optimization theory \cite{AGS05} and in probability theory \cite{ALW17}. Metric measure spaces are also used to model real-world systems and processes, for example, in image recognition \cite{Mem11}, in genetics \cite{RB19}, in machine learning \cite{ACB17}, etc.  

A natural example of a metric measure space is given by a connected graph $G$ in which all vertices and edges are labeled with numbers: the number assigned to a vertex is its ``weight'' and the number assigned to an edge is its ``length''. The corresponding metric space is formed by the set of all vertices with the shortest path metric in $G$: the distance between two vertices is the length of a shortest path between them. A measure on this space is defined on all subsets of the vertices: the measure of a subset $A$ is the total weight of all vertices of $A$. The weights and lengths can be interpreted in various ways. For example, if $G$ is a communication network, then the weight of a vertex can describe the traffic at this vertex. Another example: if $G$ is a propagation network in epidemic models, then the weight can be the number of infected individuals. 

Consider objects that can be modeled by points in metric measure spaces: suppose one object is represented by a point $x$ in a space $\mathcal{X}$ and another object is represented by a point $y$ in a space $\mathcal{Y}$. How similar are these objects? How can we measure similarity between them if the only information we have is the pairs $(\mathcal{X},x)$ and $(\mathcal{Y},y)$? Such pairs are called \emph{rooted metric measure spaces} or \emph{rooted mm spaces} for short, see Section~\ref{sec:3} for precise definitions. In this paper we address the question of similarity between objects modeled by rooted mm spaces.

The most obvious type of similarity between $(\mathcal{X},x)$ and $(\mathcal{Y},y)$ is an \emph{isomorphism} between them, which means that there is a bijection from $\mathcal{X}$ to $\mathcal{Y}$ that maps $x$ to $y$ and preserves the metric and measure. This is an ``all or nothing'' measure of similarity: any two rooted mm spaces are either similar or not. Clearly, this measure is not a good solution for applications because real-world objects, like social, biological, or technological networks, are very rarely, if ever, isomorphic to one another.

Can we improve the isomorphism-based approach to make it more flexible? How could we measure to what extent $(\mathcal{X},x)$ and $(\mathcal{Y},y)$ look isomorphic? The concept of ``approximate isomorphism'' between rooted mm spaces can be implemented using the idea proposed by Edwards \cite{Edw75} and Gromov \cite{Gro81}. To compare $(\mathcal{X},x)$ and $(\mathcal{Y},y)$, we embed $\mathcal{X}$ and $\mathcal{Y}$ into another metric measure space $\mathcal{Z}$ and compare their images in $\mathcal{Z}$. More exactly, we take embeddings $f$ and $g$ that preserve the metric and measure and compare the images $f(\mathcal{X})$ and $g(\mathcal{Y})$ in $\mathcal{Z}$. We consider $(\mathcal{X},x)$ and $(\mathcal{Y},y)$ \emph{similar} if their images are close to each other in the following sense: 
\begin{itemize}
\item the point $f(x)$ is close to the point $g(y)$ in the space $\mathcal{Z}$;
\item the set of points of $f(\mathcal{X})$ is close to the set of points of $g(\mathcal{Y})$ in the space $\mathcal{Z}$;
\item the measures induced by $f$ and $g$ in $\mathcal{Z}$ are close to one another.  
\end{itemize}
The second condition is formalized using the \emph{Hausdorff distance} and the third condition is formalized using the \emph{L\'{e}vy-Prokhorov distance}, see Section~\ref{sec:3}. Taking the infimum over all possible spaces $\mathcal{Z}$ and embeddings $f$ and $g$, we obtain a distance function on rooted mm spaces called the \emph{Gromov-Hausdorff-Prokhorov distance} (the \emph{GHP distance} for short). 

Both the isomorphism-based similarity measure and the GHP-based similarity measure have the following disadvantage for applications. Most real-world systems have the distance decay effect, also called the gravity model, which is often expressed as ``all things are related, but near things are more related than far things''. For example, when comparing points $x$ and $y$ in metric spaces, the role of their local neighborhoods is more important than the role of points that are far away from $x$ and $y$. However, neither the isomorphism approach nor the GHP distance capture this effect: all points are considered equally important, independently of their distance from $x$ and $y$. 

This disadvantage is eliminated using the distance defined in \cite{ADH13}. Loosely speaking, this distance between two rooted mm spaces combines the GHP distance with an exponential decay: a point is taken into account with a weight that exponentially decreases with increasing its distance from the root. We call it the \emph{neighborhood-based} distance and describe it in Section~\ref{sec:3}. 

Intuitively, the neighborhood-based distance is the best approach to capture similarity between points in metric measure spaces. To put this distance to work in practical applications, we need to compute it efficiently. However, it is $\np$-hard to compute the neighborhood-based distance, which follows from \cite{MSW19}. Moreover, under standard complexity-theoretic assumptions, it is not possible to approximate it with a reasonable factor in polynomial time \cite{Sch17}. 

In Section~\ref{sec:4}, we propose a tractable alternative to the neighborhood-based distance: namely, we define the \emph{radial distribution distance} between rooted mm spaces. This distance can be viewed as a coarser variant of the neighborhood-based distance or, more exactly, as a lower bound on the neighborhood-based distance. The advantage of the radial distribution distance is that it can be computed efficiently: a straightforward algorithm that computes this distance between finite rooted mm spaces takes time quasilinear in the total number of points. 

What is the idea of radial distribution distance? First, we view a point $x$ in a metric measure space $\mathcal{X}$ as the center of a ball of radius $r$ around $x$. This ball has its own measure (its ``mass'') and we consider how such masses change with increasing $r$. Second, we consider this change of masses with an exponential distance decay, which means that the ``contribution'' of points exponentially decreases with increasing their distance from $x$. The radial distribution distance between $(\mathcal{X},x)$ and $(\mathcal{Y},y)$ is basically the distance between two functions that describe the change of masses around $x$ in $\mathcal{X}$ and around $y$ in $\mathcal{Y}$. 




%
%
%
\section{Neighborhood-Based Distance}
\label{sec:3}

A metric measure space is usually defined as a complete separable metric space with a Borel measure on this metric space. However, in this paper, we deal with only compact metric spaces and finite measures. Therefore, to simplify terminology, we use the term ``metric measure space'' to refer to this restricted case.

\begin{definition}[mm space]
Let $X$ be a set and $d$ be a metric on $X$ such that $(X,d)$ is a compact metric space. Let $\mu$ be a Borel measure on this metric space. The triplet $(X,d,\mu)$ is called a \emph{metric measure space} (an \emph{mm space} for short). 
\end{definition}

We use the following notation for balls in metric spaces: for every number $r \in [0,\infty)$ and every point $x \in X$,
\begin{itemize}
\item $B_r(x) = \{y \in X \ | \ d(x,y) < r\}$ is an \emph{open ball} of radius $r$ around $x$;
\item $\overline{B}_r(x) = \{y \in X \ | \ d(x,y) \le r\}$ is a \emph{closed ball} of radius $r$ around $x$.
\end{itemize}
For every number $r>0$ and every nonempty subset $S \subseteq X$, the \emph{$r$-neighborhood} of $S$, denoted by $N_r(S)$, is the union of open balls of radius $r$ around the points of $S$. Let $S_1$ and $S_2$ be nonempty subsets of $X$. The \emph{Hausdorff distance} between them is defined by 
$$
d_H(S_1,S_2) = \inf \{r \in [0,\infty) \ | \ \mbox{$S_1 \subseteq N_r(S_2)$ and $S_2 \subseteq N_r(S_1)$}\} .
$$
Given two subsets of points in a metric space, the Hausdorff distance shows how close they are. Given two measures on a metric space, how close are they? This question is answered by the \emph{L\'{e}vy-Prokhorov distance}, a measure-theoretic analogue of the Hausdorff distance, defined as follows. Let $\mu_1$ and $\mu_2$ be finite Borel measures on a metric space $(X,d)$. Let $\mathcal{B}$ be the Borel $\sigma$-algebra on $(X,d)$. The \emph{L\'{e}vy-Prokhorov distance} (sometimes called the \emph{Prokhorov distance}) between these measures, denoted by $\pi$, is given by
$$
\begin{array}{lclcl}
\pi(\mu_1,\mu_2) & = & \inf \{r \in [0,\infty) & | & \mbox{$\mu_1(S) \le \mu_2(N_r(S)) + r$ and} \\
                 & & & & \mbox{$\mu_2(S) \le \mu_1(N_r(S)) + r$ \ for all $S \in \mathcal{B}$}\} .
\end{array}
$$


Measuring similarity between points in different mm spaces or in the same space, it is convenient to deal with \emph{rooted mm spaces} (sometimes called \emph{pointed mm spaces}). 

\begin{definition}[rooted mm space]
A \emph{rooted mm space} is a pair $(\mathcal{X},o)$ where $\mathcal{X}$ is an mm space and $o$ is a designated point in $\mathcal{X}$ called the \emph{origin} of $\mathcal{X}$. 
\end{definition}

A general approach to comparing rooted mm spaces was outlined in Section~\ref{sec:intro}. This approach combines the idea of Gromov-Hausdorff distance with the idea of L\'{e}vy-Prokhorov distance. The combination was introduced in \cite{Mie09} and has slight variations. The following definition is the version from \cite{Lei19}.

\begin{definition}[GHP distance]
Let $(\mathcal{X}_1,o_1)$ and $(\mathcal{X}_2,o_2)$ be rooted mm spaces with $\mathcal{X}_1=(X_1,d_1,\mu_1)$ and $\mathcal{X}_2=(X_2,d_2,\mu_2)$. The \emph{Gromov-Hausdorff-Prokhorov distance} (the \emph{GHP distance} for short), denoted $d_{GHP}$, is defined by
$$
\begin{array}{l}
d_{GHP}((\mathcal{X}_1,o_1),(\mathcal{X}_2,o_2)) = \\
\hspace{5mm} \inf_{\mathcal{Y},f_1,f_2} 
\{d(o_1,o_2) + d_H(f_1(X_1),f_2(X_2)) + \pi\left(\mu_1 \circ f_1^{-1}, \mu_2 \circ f_2^{-1}\right)\}
\end{array}
$$
where the infimum is taken over all mm spaces $\mathcal{Y}$ and all measurable isometries $f_1: \mathcal{X}_1 \to \mathcal{Y}$ and $f_2: \mathcal{X}_2 \to \mathcal{Y}$. The distances $d$, $d_H$, and $\pi$ denote respectively the distance, the Hausdorff distance, and the L\'{e}vy-Prokhorov distance in $\mathcal{Y}$. The measures $\mu_1 \circ f_1^{-1}$ and $\mu_2 \circ f_2^{-1}$ are the push-forward measures.
\end{definition}


As noted in Section~\ref{sec:intro}, the GHP distance has the disadvantage that $d_{GHP}$ does not capture the distance decay effect occurring in most real-world systems. This disadvantage is eliminated in the distance defined in \cite{ADH13}. We define a simplified version of this distance below and call it \emph{neighborhood-based distance}. Its idea can be informally described as follows. When comparing rooted mm spaces $(\mathcal{X}_1,o_1)$ and $(\mathcal{X}_2,o_2)$, we consider the restrictions of $\mathcal{X}_1$ and $\mathcal{X}_2$ to closed balls of radius $r$ around $o_1$ and $o_2$. For every radius $r$, we consider the GHP distance between the corresponding restrictions and sum up these distances over all values of $r$ with an exponential decrease when $r$ increases. 

Let $\overline{B}$ be a closed ball in an mm space $\mathcal{X}=(X,d,\mu)$. This ball, along with the metric and measure obtained by restricting $d$ and $\mu$ to $\overline{B}$, form the mm space called the \emph{restriction} of $\mathcal{X}$ to $\overline{B}$. 

\begin{definition}[neighborhood-based distance]
\label{def:nbd}
Let $(\mathcal{X}_1,o_1)$ and $(\mathcal{X}_2,o_2)$ be rooted mm spaces. For every number $r \in [0,\infty)$, let $\delta_r$ denote the GHP distance between the restriction of $\mathcal{X}_1$ to $\overline{B}_r(o_1)$ and the restriction of $\mathcal{X}_2$ to $\overline{B}_r(o_2)$. The \emph{neighborhood-based distance}, denoted $d_{nb}$, is defined by
$$
d_{nb}((\mathcal{X}_1,o_1),(\mathcal{X}_2,o_2)) = \int_0^\infty e^{-r} \, \delta_r \ dr
$$
\end{definition}

It follows from the definition that 
$$
d_{nb}((\mathcal{X}_1,o_1),(\mathcal{X}_2,o_2))=0
$$
for isomorphic $(\mathcal{X}_1,o_1)$ and $(\mathcal{X}_2,o_2)$ even if they are different. Therefore, $d_{nb}$ is not a metric but, as shown in \cite{ADH13}, $d_{nb}$ is a pseudometric.


%
%
%
\section{Radial Distribution Distance}
\label{sec:4}

Can we compute the neighborhood-based distance $d_{nb}$ efficiently? Note that an efficient algorithm for computing $d_{nb}$ could also be used to compute the Gromov-Hausdorff distance efficiently. However, as shown in \cite{Mem07,MSW19}, it is $\np$-hard to compute the Gromov-Hausdorff distance for finite metric spaces (see Proposition 3.16 in \cite{MSW19}). Moreover, under standard complexity-theoretic assumptions, there is no polynomial-time approximation algorithm with a reasonable factor for computing this distance \cite{Sch17}. 

In this section, we define another distance between rooted mm spaces called the \emph{radial distribution distance} and denoted by $d_{rd}$. On the one hand, $d_{rd}$ is coarser than $d_{nb}$, more exactly, $d_{rd}$ is a lower bound on $d_{nb}$. On the other hand, $d_{rd}$ can be computed efficiently: computing the radial distribution distance between finite rooted mm spaces takes time quasilinear in the total number of points. 

\subsection{Definition and Properties}

Consider rooted mm spaces $(\mathcal{X}_1,o_1)$ and $(\mathcal{X}_2,o_2)$ where $\mathcal{X}_1 = (X_1,d_1,\mu_1)$ and $\mathcal{X}_2 = (X_2,d_2,\mu_2)$. To define the \emph{radial distribution distance} between them, we first define the following functions $m_1$ and $m_2$ from $[0,\infty)$ to itself: for every number $r \in [0,\infty)$,
$$
\begin{array}{lcl}
m_1(r) & = & \mu_1\left(\overline{B}_r(o_1)\right) = \mu_1\left(\{x \in X_1 \ | \ d_1(x,o_1) \le r\}\right) \\
m_2(r) & = & \mu_2\left(\overline{B}_r(o_2)\right) = \mu_2\left(\{x \in X_2 \ | \ d_2(x,o_2) \le r\}\right) 
\end{array}
$$
That is, $m_1(r)$ is the measure (we could call it ``mass'' or ``weight'') of the ball of radius $r$ around the origin $o_1$ and, similarly, for $m_2(r)$. The functions are non-decreasing and bounded. Also, $m_1$ and $m_2$ are c\`{a}dl\`{a}g functions, i.e., they are right continuous with left limits, see Lemma 2.8 in \cite{ADH13}. Therefore, the distance function in the definition below is well defined.


\begin{definition}[radial distribution distance]
Let $(\mathcal{X}_1,o_1)$ and $(\mathcal{X}_2,o_2)$ be rooted mm spaces. The \emph{radial distribution distance}, denoted $d_{rd}$, is defined by
$$
d_{rd}((\mathcal{X}_1,o_1),(\mathcal{X}_2,o_2)) = \int_0^\infty e^{-r} \, |m_1(r)-m_2(r)| \, dr
$$
\end{definition}


We use the term ``radial distribution distance'' keeping in mind radial distribution functions from physics where they are used to measure the probability of finding a particle at distance $r$ from a certain ``origin'' particle. In the setting of metric measure spaces, a radial distribution function describes how the total ``mass'' (``weight'') of points at distance $r$ from a center $x$ changes when $r$ increases. The radial distribution distance is basically a combination of the $L^1$ distance between the cumulative versions of the radial distribution functions and an exponential distance decay. Theorems~\ref{th:pseudometric-rd} and \ref{th:lower-bound} show basic properties of the radial distribution distance. 


\begin{theorem}
\label{th:pseudometric-rd}
The function $d_{rd}$ is a pseudometric on the set of rooted mm spaces. 
\end{theorem}

\begin{proof}
It is obvious that $d_{rd}$ is a distance function. The triangle inequality for $d_{rd}$ can be seen from the following two facts:
\begin{itemize}
\item The functions $e^{-r} m_1(r)$ and $e^{-r} m_2(r)$ are measurable functions on $[0,\infty)$ such that the integrals $\int_0^\infty e^{-r} \, m_1(r) \, dx$ and $\int_0^\infty e^{-r} \, m_2(r) \, dx$ are finite. 
\item The distance function 
$$
d(f_1,f_2) = \int_0^\infty |f_1(x)-f_2(x)| \, dx
$$
is the $L^1$ metric on the set of measurable functions $f$ such that the integral $\int_0^\infty |f(x)| \, dx$ is finite. 
\end{itemize}
\qed
\end{proof}


\begin{theorem}
\label{th:lower-bound}
For all rooted mm spaces $(\mathcal{X}_1,o_1)$ and $(\mathcal{X}_2,o_2)$, 
\begin{equation}
\label{eq:lower-bound}
d_{rd}((\mathcal{X}_1,o_1),(\mathcal{X}_2,o_2)) \, \le \, d_{nb}((\mathcal{X}_1,o_1),(\mathcal{X}_2,o_2)) .
\end{equation}
\end{theorem}

\begin{proof}
Let $\overline{B}_r(o_1)$, $\overline{B}_r(o_2)$, and $\delta_r$ be the same as in Definition~\ref{def:nbd}. Let $\mu_1$ and $\mu_2$ be the measures in $\mathcal{X}_1$ and $\mathcal{X}_2$ respectively. We prove the following inequality that implies claim (\ref{eq:lower-bound}): 
\begin{equation}
\label{eq:stronger-inequality}
|\mu_1(\overline{B}_r(o_1)) - \mu_2(\overline{B}_r(o_2))| \le \delta_r.
\end{equation}
By the definition of the GHP distance, we have 
\begin{equation}
\label{eq:delta-pi}
\delta_r \ge \pi(\mu_1 \circ f_1^{-1}, \mu_2 \circ f_2^{-1}) 
\end{equation}
where $f_1$ and $f_2$ are arbitrary measurable isometries of the corresponding restrictions to an arbitrary mm space $\mathcal{Y}$ and $\pi$ is the L\'{e}vy-Prokhorov distance in $\mathcal{Y}$. By the definition of $\pi$, 
\begin{equation}
\label{eq:pi-epsilon}
\pi(\mu_1 \circ f_1^{-1}, \mu_2 \circ f_2^{-1}) = \gamma
\end{equation}
where $\gamma$ is the infimum of all $\epsilon \ge 0$ such that 
$$
\left\{
\begin{array}{lcl}
\mu_1\left(\overline{B}_r(o_1)\right) & \le & \mu_2\left(\overline{B}_{r+\epsilon}(o_2)\right) + \epsilon \\
\mu_2\left(\overline{B}_r(o_2)\right) & \le & \mu_1\left(\overline{B}_{r+\epsilon}(o_1)\right) + \epsilon
\end{array}
\right.
$$
Both inequalities above hold if we take 
$$
\epsilon = |\mu_1(\overline{B}_r(o_1)) - \mu_2(\overline{B}_r(o_2))|
$$ 
and, therefore, we have 
\begin{equation}
\label{eq:gamma}
\gamma \ge |\mu_1(\overline{B}_r(o_1)) - \mu_2(\overline{B}_r(o_2))|.
\end{equation}
Now, combining (\ref{eq:delta-pi})--(\ref{eq:gamma}), we obtain inequality (\ref{eq:stronger-inequality}).
\qed
\end{proof}


The theorem above shows that $d_{rd}$ is a lower bound on $d_{nb}$. This bound is strict, which can be seen from the following simple example. Consider a rooted mm space on a set of three points: $\{a, b, c\}$ where $a$ is the origin. The distance between any two points is $1$. The measure assigned to each point is $1$. Consider another rooted mm space that differs from the first one in only the measure of $b$ and $c$: in the second space, the measure of $b$ is $0.5$ and the measure of $c$ is $1.5$. It is easy to see that the radial distribution distance between these spaces is $0$, while the neighborhood distance is not zero.


\subsection{Computing the Radial Distribution Distance}

We describe a straightforward algorithm that computes the radial distribution distance between finite rooted mm spaces efficiently. In the description, we ignore all problems of numerical analysis and deal with real numbers. The algorithm is illustrated by applying it to the following example of two rooted mm spaces $(\mathcal{X}_1,o_1)$ and $(\mathcal{X}_2,o_2)$. 


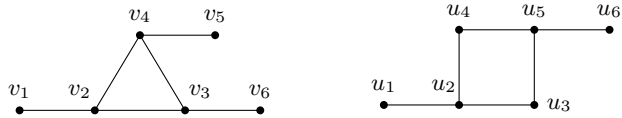
\begin{figure}
\begin{center}
  \begin{minipage}{0.45\linewidth}
\begin{flushright}
\begin{tikzpicture}
\tikzset{dot/.style = {circle, fill, minimum size=#1,
              inner sep=0pt, outer sep=0pt},
dot/.default = 3pt 
}
\draw (0,0) node[dot,fill=black,label=above:$v_1$]{};
\draw (1,0) node[dot,fill=black,label={[xshift=-0.2cm]$v_2$}]{};
\draw (2.2,0) node[dot,fill=black,label={[xshift=0.2cm]$v_3$}]{};
\draw (1.6,1) node[dot,fill=black,label=above:$v_4$]{};
\draw (2.6,1) node[dot,fill=black,label=above:$v_5$]{};
\draw (3.2,0) node[dot,fill=black,label=above:$v_6$]{};
\draw (0,0) -- (1,0) -- (2.2,0)-- (1.6,1) -- (1,0);
\draw (2.2,0) -- (3.2,0);
\draw (1.6,1) -- (2.6,1);
\end{tikzpicture}
\end{flushright}
\end{minipage} \hspace{10mm}\noindent\begin{minipage}{0.45\linewidth}
          \begin{tikzpicture}
\tikzset{
dot/.style = {circle, fill, minimum size=#1,
              inner sep=0pt, outer sep=0pt},
dot/.default = 3pt 
}
\draw (0,0) node[dot,fill=black,label=above:$u_1$]{};
\draw (1,0) node[dot,fill=black,label={[xshift=-0.2cm]$u_2$}]{};
\draw (2,0) node[dot,fill=black,label=right:$u_3$]{};
\draw (1,1) node[dot,fill=black,label=above:$u_4$]{};
\draw (2,1) node[dot,fill=black,label=above:$u_5$]{};
\draw (3,1) node[dot,fill=black,label=above:$u_6$]{};
\draw (0,0) -- (1,0) -- (2,0)-- (2,1) -- (1,1) -- (1,0);
\draw (2,1) -- (3,1);
\end{tikzpicture}
\end{minipage}
\end{center}
\caption{Graphs for $(\mathcal{X}_1,o_1)$ and $(\mathcal{X}_2,o_2)$}
\label{fig:2rs}
\end{figure}


\paragraph{Example.}
We define $(\mathcal{X}_1,o_1)$, where $\mathcal{X}_1=(X_1,d_1,\mu_1)$, using the graph on the left in Fig.~\ref{fig:2rs}: namely, $X_1 = \{v_1, \ldots, v_6\}$ and $d_1(v_i,v_j)$ is the length of a shortest path between $v_i$ and $v_j$. The measure $\mu_1$ is defined by assigning numerical values to all one-element subsets of $X_1$; this assignment is shown in the second column of the table in Fig.~\ref{fig:input}. The origin $o_1$ is defined to be $v_1$. 

The rooted mm space $(\mathcal{X}_2,o_2)$, where $\mathcal{X}_2=(X_2,d_2,\mu_2)$, is defined similarly using the graph on the right in Fig.~\ref{fig:2rs}. The underlying set $X_2$ consists of the vertices $u_1, \ldots, u_6$ with the shortest-path metric. The measure $\mu_2$ is given by the forth column of the table in Fig.~\ref{fig:input} and the origin $o_2$ is defined to be $u_1$. 

\paragraph{Algorithm.}
The radial distribution distance between finite rooted mm spaces can be computed using only part of information about the input spaces. The algorithm uses the distances between the origins and all other points; it does not use the distances between non-origin points. In terms of matrices, if a distance function in a rooted mm space with $n$ points is represented by an $n \times n$ matrix, then the algorithm uses only one row (or column) corresponding to the origin. Therefore, we assume that the algorithm takes as input only the following:
\begin{itemize}
\item the distances between the origins and non-origins;
\item the measures of all one-element sets.
\end{itemize}
When the algorithm is applied to our example, it takes as input the table in Fig.~\ref{fig:input}. 


\begin{figure}
\begin{center}
\begin{tabular}{||c|c||c|c||} 
\multicolumn{2}{c}{rooted mm space $(\mathcal{X}_1,o_1)$} & 
\multicolumn{2}{c}{rooted mm space $(\mathcal{X}_2,o_2)$} \\ \hline
distances        & measures           & distances        & measures  \ \\ \hline
\ $d_1(v_1,v_1)=0$ \ & \ $\mu_1(\{v_1\})=1$ \ & \ $d_2(u_1,u_1)=0$ \ & \ $\mu_2(\{u_1\})=1$ \ \\
\ $d_1(v_1,v_2)=1$ \ & \ $\mu_1(\{v_2\})=1$ \ & \ $d_2(u_1,u_2)=1$ \ & \ $\mu_2(\{u_2\})=1$ \ \\
\ $d_1(v_1,v_3)=2$ \ & \ $\mu_1(\{v_3\})=2$ \ & \ $d_2(u_1,u_3)=2$ \ & \ $\mu_2(\{u_3\})=2$ \ \\
\ $d_1(v_1,v_4)=2$ \ & \ $\mu_1(\{v_4\})=2$ \ & \ $d_2(u_1,u_4)=2$ \ & \ $\mu_2(\{u_4\})=2$ \ \\
\ $d_1(v_1,v_5)=3$ \ & \ $\mu_1(\{v_5\})=1$ \ & \ $d_2(u_1,u_5)=3$ \ & \ $\mu_2(\{u_5\})=1$ \ \\
\ $d_1(v_1,v_6)=3$ \ & \ $\mu_1(\{v_6\})=1$ \ & \ $d_2(u_1,u_6)=4$ \ & \ $\mu_2(\{u_6\})=1$ \ \\ \hline
\end{tabular}
\end{center}
\caption{Input for the algorithm}
\label{fig:input}
\end{figure}


Here is a sketch of the algorithm:

\begin{enumerate}
\item Compute the set $R_1$ of all values of radius in $(\mathcal{X}_1,o_1)$: 
$$
R_1 \ = \ \{r \ | \ \mbox{$r=d_1(o_1,x_1)$ for some point $x_1$ in $\mathcal{X}_1$}\}
$$ 
where $d_1$ is the metric in $\mathcal{X}_1$. 
\item Compute the set $R_2$ of all values of radius in $(\mathcal{X}_2,o_2)$ defined similarly. 
\item Compute the union $R = R_1 \cup R_2$. In our example, $R=\{0, 1, 2, 3, 4\}$.
\item Compute the \emph{cumulative radial distribution} of $(\mathcal{X}_1,o_1)$, i.e., the set of pairs $(r,m_1(r))$ where $r \in S$ and $m_1(r) = \mu_1\left(\overline{B}_r(o_1)\right)$. In our example, this set is
$$
\{(r,m_1(r)\}_{r \in R} = \{(0,1), (1,2), (2,6), (3,8), (4,8)\}.
$$
\item Compute the cumulative radial distribution of $(\mathcal{X}_2,o_2)$ defined similarly. In our example, this set is 
$$
\{(r,m_2(r)\}_{r \in R} = \{(0,1), (1,2), (2,6), (3,7), (4,8)\}.
$$
\item Compute the radial distribution distance by
$$
d_{rd}((\mathcal{X}_1,o_1),(\mathcal{X}_2,o_2)) = \sum_{r \in R} e^{-r} |m_1(r)-m_2(r)| .
$$
Applying this to our example, we obtain 
$$
d_{rd}((\mathcal{X}_1,o_1),(\mathcal{X}_2,o_2)) = e^{-3}\cdot 1 \approx 0.05 .
$$
\end{enumerate}
Note that the algorithm takes $O(n \log n)$ steps where $n$ is the number of points in the input spaces.

Considering the point $v_1$ in $\mathcal{X}_1$, what points in $\mathcal{X}_2$ are most ``similar'' to $v_1$ and what points are most ``dissimilar'' to it? The calculations show that $u_1$ and $u_6$ are most similar to $v_1$ (with $d_{rd} \approx 0.05$), while the points $u_3$ and $u_4$ are most dissimilar to $v_1$ (with $d_{rd} \approx 2.01$). 


\subsection{Extension for Feature Vectors}

The radial distribution distance $d_{rd}$ can be used to measure similarity between points in a metric space if points are described with a single feature: a value of this feature for a given point is viewed as the point's ``weight''. However, it is more typical that a point is described by a feature vector rather than a single feature. Each component of the vector corresponds to a measure on the metric space. For example, consider a recommender system for movies that uses item-item collaborative filtering. A metric on movies is based on the similarity between them calculated using people's ratings. In addition to the metric, each movie is described by a feature vector that can include, for example, the number of reviews, budget, box office, etc. \cite{NDK15}   

How can we compare points in a metric space if points are described using feature vectors? Suppose a feature vector consists of $k$ components that correspond to measures $\mu_1,\ldots,\mu_k$. Each measure $\mu_i$ determines the radial distribution distance $d_{rd}^{(i)}$ and we can consider their sum
\begin{equation}
\label{eq:sum}
d_{rd}^* = d_{rd}^{(1)} + \ldots + d_{rd}^{(k)} .
\end{equation}
The value $d_{rd}^*((\mathcal{X},x),(\mathcal{Y},y))$ essentially shows the distance between the neighborhoods of $x$ and $y$ if we compare points by their feature vectors. Note that the sum of pseudometrics is also a pseudometric. Also note that instead of the sum in (\ref{eq:sum}), we could take any other norm, for example, the Euclidean norm or the maximum.


%
%
%

%
%
\end{document}